\documentclass[letterpaper, 10 pt, journal]{IEEEtran}
\usepackage{amssymb}
\usepackage{amsmath}
\usepackage{graphicx}
\usepackage{amssymb}
\usepackage{amsthm}
\usepackage{bbm}
\usepackage{easylist}
\usepackage{fp}
\usepackage{siunitx}
\usepackage{xcolor}
\usepackage{color}
\usepackage{multirow}
\usepackage{subfigure}
\usepackage{ctable}
\usepackage{array}
\usepackage{mathabx}

\usepackage{tabu}

\usepackage{soul}

\usepackage[ruled]{algorithm2e}

\newtheorem{theorem}{Theorem}

\newtheorem{corollary}{Corollary}

\newcommand{\edgeki}[2]{   \ensuremath{ e(#1,#2)}  } 
\newcommand{\nodeki}[2]{   \ensuremath{ v(#1,#2)}  }

\newcommand{\waitingtime}[2]{ \ensuremath{   w^{#1}_{#2}   }}

\newcommand{\actioni}[1]{   \ensuremath{ \boldsymbol w ^{#1}}  }

\newcommand{\actionspace}[1]{   \ensuremath{ \mathcal W ^{#1}}  }

\newcommand{\DT}[2]{   \ensuremath{ t_{#1}^{#2}}  }

\newcommand{\TT}[2]{\ensuremath{ \tau(#1,#2)}}

\newcommand{\RTT}[2]{\ensuremath{ \boldsymbol\tau(#1,#2)}}

\newcommand{\C}[4]{\ensuremath{ \mathcal C\left(#1,#2,#3,#4\right)}} 

\newcommand{\missonstart}[2]{   \ensuremath{ \tau_{#1}^{#2}}  }

\newcommand{\ES}[1]{\ensuremath{\mathsf{E}\left[#1 \right]}}
\newcommand{\CES}[2]{\ensuremath{\mathsf{E}\left[\left. #1\right| #2\right]}}

\newcommand{\DDT}[3]{\ensuremath{ t^{#2}_{#1\left| #3\right.}}}

\newcommand{\WT}[3]{\ensuremath{ w^{#2}_{#1\left| #3\right.}}} 

\newcommand{\AWT}[2]{\ensuremath{ \boldsymbol{w}^{#1 }_{ #2}}} 

\newcommand{\MC}[5]{\ensuremath{  {\mathcal C}_{#5}\left(#1,#2,#3,#4\right)}} 

\newcommand{\actionipotp}[1]{   \ensuremath{ \boldsymbol{ \widehat{w}}^{#1}  }  }
\newcommand{\actionipotpp}[1]{   \ensuremath{ \boldsymbol{ \widecheck{w}}^{#1}  }  }

\newcommand{\waitingtimep}[2]{   \ensuremath{  \widehat{w}}^{#1}_{#2}  }
\newcommand{\waitingtimepp}[2]{   \ensuremath{  \widecheck{w}}^{#1}_{#2}  }

\newcommand{\DTp}[2]{   \ensuremath{  \hat{t}_{#2}^{#1} }  } 
\newcommand{\DTpp}[2]{   \ensuremath{  \check{t}_{#2}^{#1}  }  } 

\newcommand{\NE}[1]{   \ensuremath{ \boldsymbol{ \ring{w}}^{#1}  }  }

\begin{document}

\IEEEoverridecommandlockouts

\title{Strategic Hub-Based Platoon Coordination under Uncertain Travel Times }

\author{\IEEEauthorblockN{Alexander Johansson, Ehsan Nekouei, Karl Henrik Johansson and Jonas M\aa rtensson }	\thanks{This work is supported by the Strategic Vehicle Research and Innovation Programme through the project Sweden for Platooning, Horizon 2020 through the project Ensemble, the Knut and Alice Wallenberg Foundation, the Swedish Foundation for Strategic Research and the Swedish Research Council. } 
	\thanks{A. Johansson, K. H. Johansson and J. M\aa rtensson are with the Integrated Transport Research Lab and Division of Decision and Control,
		School of Electrical Engineering and Computer Science, KTH Royal Institute
		of Technology, Stockholm, Sweden.,
		SE-100 44 Stockholm, Sweden. Emails:
		{\tt\small \{alexjoha, kallej, jonas1\}@kth.se}}
	\thanks{E. Nekouei are with the Department of Electrical Engineering,
	 City University of Hong Kong, Hong Kong. 
	 Email: {\tt \small \{enekouei\}@cityu.edu.hk}  }}
	
\maketitle

\thispagestyle{plain}
\pagestyle{plain}

\begin{abstract}

We study the strategic interaction among vehicles in a non-cooperative platoon coordination game. Vehicles have predefined routes in a transportation network with a set of hubs where vehicles can wait for other vehicles to form platoons. Vehicles decide on their waiting times at hubs and the utility function of each vehicle includes both the benefit from platooning and the cost of waiting. We show that the platoon coordination game is a potential game when the travel times are either deterministic or stochastic, and the vehicles decide on their waiting times at the beginning of their journeys. We also propose two feedback solutions for the coordination problem when the travel times are stochastic and vehicles are allowed to update their strategies along their routes.  The solutions are evaluated in a simulation study over the Swedish road network. It is shown that uncertainty in travel times affects the total benefit of platooning drastically and the benefit from platooning in the system increases significantly when utilizing feedback solutions.

\end{abstract}

\section{Introduction}

\IEEEPARstart{V}{ehicle} platoon refers to a group of vehicles that drive in a formation with small inter-vehicular distances. The lead vehicle in a platoon, called the platoon leader, is typically maneuvered by a human driver. The other vehicles, called the platoon followers, automatically follow their respectively in-front-driving vehicles. We refer the interested reader to \cite{Horowitz2000} and \cite{Besselink2016} for overviews of control designs and system architectures for vehicle platooning. Platooning has the potential to be a substantial element in the future intelligent transportation system thanks to the following benefits:
\begin{itemize}
	\item Decreased workload of the drivers in the follower vehicles. This can potentially lead to enormous savings if the drivers can utilize their time to perform administrative duties or if the platoon followers can be unmanned. According to a recent estimate \cite{Chottani2018}, the total cost of ownership of trucks in the US may decrease by $10\%$ in the period 2022--2027 due to platooning with unmanned followers.
	\item Increased road capacity and safer driving thanks to synchronized driving and shared information between vehicles in platoons, \emph{e.g.,} sharing vehicle parameters, sudden accelerations, road conditions, surrounding traffic, etc. Increased road capacity and safety were demonstrated in numerical simulations in \cite{Ioannou1993} and \cite{Fernandes2012}. The simulation study in \cite{Jo2019} over the Korean transport network demonstrated capacity improvements and decreased travel times.
	\item Reduced fuel consumption thanks to reduced air drag. This was demonstrated in numerical studies in \cite{Davila2013}, \cite{Bishop2017} and by field experiments in \cite{Alam2015}, \cite{Browand2004}, \cite{Tsugawa2016}, where potential energy savings of around $10 \%$ were reported.
\end{itemize}

Vehicles need coordination in order to meet in the transportation network to form platoons. In this paper, we consider the coordination problem illustrated in Figure \ref{largenetwork2}, where vehicles can wait and form platoons at certain locations, called hubs. Examples of hubs in today's transportation infrastructure are  freight terminals, gas stations, parking places, tolling stations and harbors. The rest time of drivers is strictly regulated and long-distance drivers are forced to rest with a certain regularity. Resting places are ideal hub locations since the  drivers can rest while waiting for other vehicles to platoon with. An alternative to forming platoons at hubs is to form platoons on the road, without stopping at hubs. Then, a platoon can be formed if some vehicles speed up or some vehicles slow down. The main drawback of forming platoons on the road is that vehicles that slow down may decrease the traffic flow and vehicles that speed up may violate the speed limits.

\begin{figure}
	\centering
	\includegraphics[scale=0.44]{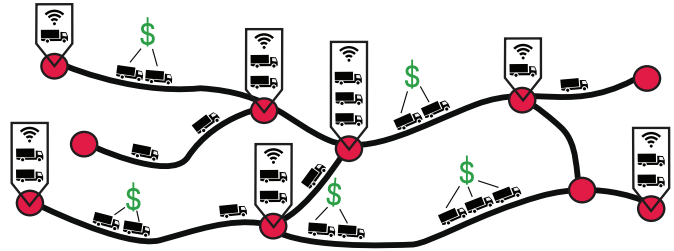}
	\caption{Vehicles at hubs (nodes) in the transportation network communicate and can decide to wait for others in order to form platoons. Vehicles that enter a road (edges) at the same time form a platoon on that road and benefit from platooning.}
	\label{largenetwork2}
\end{figure}

 Solutions to the hub-based platoon coordination problem have been proposed in \cite{Zhang2017}, \cite{Boysen2018}, \cite{Larsen2019}, where the aim was to maximize the total profit from platooning assuming all vehicles are owned by the same transportation company or vehicles have the same objective. The authors in \cite{Zhang2017} considered the platoon coordination problem of two vehicles with stochastic travel times, and in \cite{Boysen2018} and \cite{Larsen2019} it was assumed that the travel times were deterministic.

The problem of platoon coordination among  competitive transportation firms has been studied in \cite{Farokhi2013}, \cite{Sun2019} and our past research effort \cite{Johansson2018}. In \cite{Sun2019}, a socially optimal solution was proposed that maximizes the total profit from platooning on a common road with deterministic travel times, and the platooning profit is distributed between the vehicles such that vehicles have no incentive to leave their platoons. Different from \cite{Sun2019}, we analyze in this paper the strategic behavior of the vehicles, captured by the notion of Nash equilibrium~(NE), on a general road network with stochastic travel times.

Strategic platoon coordination problems wherein each vehicle seeks to maximize its own profit from platooning was considered in \cite{Farokhi2013} and \cite{Johansson2018}. The solutions in \cite{Farokhi2013} and \cite{Johansson2018} are limited to one-edge graphs and  tree graphs, respectively, where vehicles have the same origin and the travel times on all edges are deterministic. Different from these proposals, the solution in this paper holds for general graphs with arbitrary vehicle routes and stochastic travel times.

Cooperative solutions to the platoon coordination problem where vehicles slow down or speed up in order to form platoons have been studied in \cite{Liang2016}, \cite{Larsson2015}, \cite{Hoef2018}, and \cite{Xiong2019}, where the aim was to maximize the total profit from platooning for all vehicles. Different from these papers, we assume that vehicles are owned by competing transportation  companies and each vehicle is interested in optimizing its individual utility function. A review on planning strategies for platooning, including platoon coordination, is given in \cite{Bhoopalam2018}.

For the hub-based platoon coordination problem considered in this paper, vehicles decide on how long time to wait at hubs in order to maximize their individual utility functions. This is realistic if we consider vehicles that are owned by different transportation companies. The utility functions include both the benefit from platooning and the cost of waiting. Vehicles have arbitrary origins and destinations in a road network with general topology. Moreover, we consider time-varying travel times, as well as both deterministic and stochastic travel times. Figure \ref{TTasd} shows the travel time and its standard deviation on the highway between Botkyrka to V{\"a}stertorp close to Stockholm, Sweden, for vehicles departing from Botkyrka in $5$-minute intervals between 5:00 and 11:00 a.m. The data was collected during $10$ working days. According to this figure, the uncertainty of the travel time is higher during the traffic peak period compared with the off-peak periods. This motivates taking into account uncertainty in travel times especially when vehicles pass by highways near urban regions.

\begin{figure}[t!]
	\centering
	\subfigure[Travel times  ]{\includegraphics[scale=0.6]{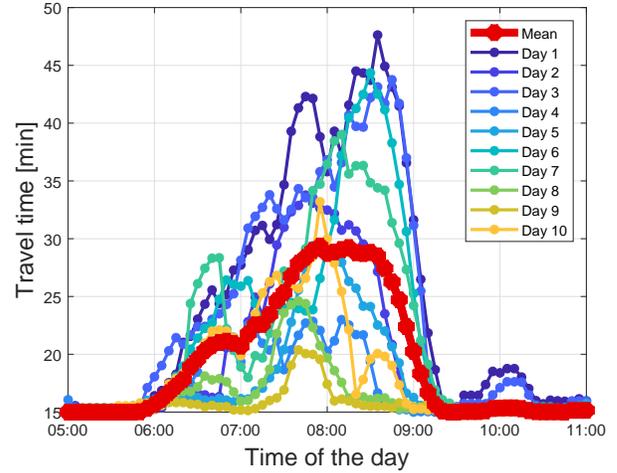}\label{TT}}
	
	\subfigure[Standard deviation ]{\includegraphics[scale=0.6]{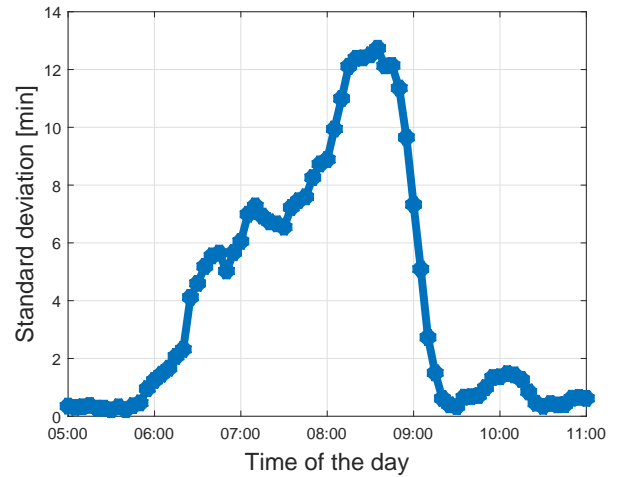}\label{TTsd}}
	
	\caption{Travel time and its standard deviation at $10$ working days between Botkyrka and V\"astertorp. The travel time during the off-peak periods is approximately $15$~min and the standard deviation is higher during the peak period.}
	\label{TTasd}
	
\end{figure}

The contributions of this paper can be summarized as follows:

\begin{itemize}
	\item We develop two game-theoretic models of the platoon coordination problem. The first game is for deterministic travel times and the second game is for stochastic travel times. 
	
	\item We show that these platoon coordination games are exact potential games and admit at least one pure Nash equilibrium~(NE). The NE is the open-loop solution to the platoon coordination problem, both for deterministic and stochastic travel times.  In the open-loop solution, vehicles calculate their waiting times initially and do not update their waiting times along their journeys.
	  
	\item To counter uncertain stochastic travel times, we propose two feedback solutions to the platoon coordination problem in which vehicles can update their decisions along their journeys.

	\item We perform a simulation study over the Swedish road network to evaluate the
	proposed platoon coordination solutions. The travel
	times on the roads in the simulation are time-varying
	and have a stochastic time-delay that is generated from real
	travel time data. Another objective of the simulation study is to investigate the potential benefits of non-cooperative platooning when vehicles are owned by different transportation  companies and aim to maximize their individual  profits from platooning.
\end{itemize}

This paper is structured as follows. The platoon coordination problem with deterministic travel times and its game formulation is considered in Section \ref{DPMGS}. In Section \ref{PMGSt1}, we consider the problem with stochastic travel times. In Section \ref{PMPS}, two feedback solutions are proposed, where the actions of vehicles are updated along with their journeys, for the stochastic setup in Section \ref{PMGSt1}. We evaluate the platoon coordination solutions in a simulation study in Section \ref{sim}. Finally, the conclusions are given in Section \ref{con}.

\section{Platoon coordination game with deterministic travel times}\label{DPMGS}

In this section, we define the platoon coordination problem when the vehicles' travel times are known \textit{a priori}. First, the system model is presented, which includes the graph representation of the road network, vehicles' waiting times at hubs, their departure times from the hubs, and their benefit from platooning. Then, a game that models the platoon coordination is presented and we show that it admits an NE.

\subsection{Graph representation of road network}\label{SM}

We consider a road network represented by a directed graph $\mathcal G=(\mathcal V,\mathcal E)$. The nodes $\mathcal V$ represent locations (hubs) where vehicles can wait in order to platoon with others. The edges $\mathcal E$ represent roads that connect the locations where vehicles can wait.  Let $\TT{e}{t}\in \mathbb{Z}_+$ denote the travel time on edge $e \in \mathcal E$ of vehicles that enters $e$ at time instance $t \in \mathbb{Z}_+$. The vehicles to be coordinated are enumerated $1$ to $N$, and the set of vehicles is denoted $\mathcal{N}~=~\{1,...,N\}$. Each vehicle has a fixed path in the road network to traverse. An example of the path of vehicle $i$ is illustrated in Figure \ref{pathostate}. The set of edges in the path of vehicle $i$ is denoted $\mathcal P^i$. The $k$th edge and $k$th node in the path of vehicle $i$ are denoted $\edgeki{i}{k}\in \mathcal E$ and $\nodeki{i}{k} \in \mathcal V$, respectively. The edge $\edgeki{i}{k}$ thus connects nodes $\nodeki{i}{k}$ and $\nodeki{i}{k+1}$. The time instance when vehicle $i$ starts its journey (arrives to its first node $\nodeki{i}{1}$) is denoted $\missonstart{0}{i}\in \mathbb{Z}_+$. With a slight abuse of notation we define $\tau$ to be the set containing both the travel times $\TT{e}{t}$, for all $e\in \mathcal E$ and all $t\in \mathbb{Z}_+$, and start times $\missonstart{0}{i}$, for all $ i\in \mathcal N$. In our setup, time is discrete and we assume that the time step length is small in comparison to the travel times between hubs.

\begin{figure}[t!]
	\centering
	\subfigure[Three nodes and two edges in the path of vehicle $i$.   ]{\includegraphics[scale=0.7]{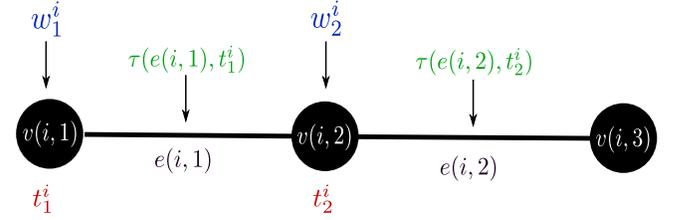}\label{path}}
	
	\subfigure[Vehicle $i$'s location as a function of time. The location is a node or an edge. The start time, waiting times, travel times and departure times are marked on the time axis.	]{\includegraphics[scale=1]{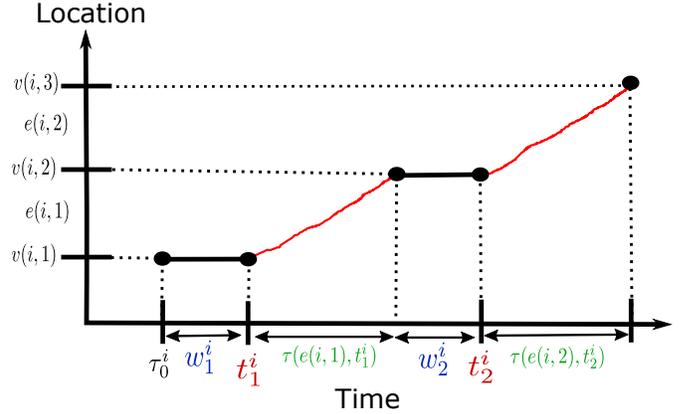}\label{state}}
	
	\caption{The path of vehicle $i$ and its location as a function of time are illustrated in (a) and (b), respectively. The waiting times at nodes (blue), the travel times on edges (green) and the departure times at nodes (red) are marked. }
	
	\label{pathostate}
	
\end{figure}

\subsection{Waiting times at nodes} \label{static}

The vehicles decide to wait at nodes in order to form platoons with others. The waiting times are the actions of the  vehicles. The waiting time of vehicle $i$ at the node $\nodeki{i}{k}$ is denoted $\waitingtime{i}{k}\in \mathbb{Z}_+$. The waiting times of vehicle $i$ at the nodes in its path is denoted  $\boldsymbol w^i~=~(\waitingtime{i}{1},...,\waitingtime{i}{|\mathcal P^i|}) $. The action space of vehicle $i$ is finite and denoted by  $\actionspace{i}$. The action space is assumed to be finite due to various restrictions such as mission deadlines or rest period restrictions. The waiting times of all vehicles are denoted by $\boldsymbol w$, and by $\actioni{-i}$ we denote the waiting times of all vehicles except vehicle $i$.

\subsection{Departure times at nodes}

The departure time of vehicle $i$ from the node $\nodeki{i}{k}$ (the $k$th node in its path) is denoted $\DT{k}{i}$. The departure time of vehicle $i$ from node $\nodeki{i}{1}$ (its first node) is
\begin{equation}\label{firstdeparturetime}
\DT{1}{i}= \missonstart{0}{i}+\waitingtime{i}{1},
\end{equation}
where $\missonstart{0}{i}$ is the time instance when vehicle $i$ arrives to node $\nodeki{i}{1}$. The departure times of vehicle $i$ from the other nodes in its path are recursively calculated as 
\begin{align}\label{arrivaltimes}
\DT{k+1}{i}=\waitingtime{i}{k+1}+\DT{k}{i}+\TT{\edgeki{i}{k}}{\DT{k}{i}}, 
\end{align}
where $\TT{\edgeki{i}{k}}{\DT{k}{i}}$ is the travel time on edge $\edgeki{i}{k}$ if it is entered at time instance $\DT{k}{i}$. Note that the departure times of vehicles from nodes depend on the elements in $\tau$ and $\boldsymbol w$, that is, the travel times on edges, the start times of vehicles, and their waiting times at nodes. From \eqref{firstdeparturetime} and \eqref{arrivaltimes}, vehicle $i$ can calculate its departure times from all nodes in its path. The departure times are used next to indicate which vehicles form platoons, at which edges the platoons are formed, and when the platoons are formed.

\subsection{Utility functions} \label{utilitysec}

If a group of vehicles departs from a node and enter the same edge at the same time instance, then they form a platoon and therefore benefit from platooning. The set of vehicles that enter edge $e$ at time instance $t$ is denoted 
\begin{equation}\label{PF}
\C{e}{t}{\boldsymbol w}{\tau}  =~   \{i\in~ \mathcal N | \edgeki{i}{k}=~e, \ \DT{k}{i}=~t \}.
\end{equation}
From \eqref{firstdeparturetime} and \eqref{arrivaltimes}, it is seen that $\C{e}{t}{\boldsymbol w}{\tau}$ depends on $\tau$ and $\boldsymbol w$. The set $\C{e}{t}{\boldsymbol w}{\tau}$ is important since it indicates the vehicles that start platooning on edge $e$ at time instance $t$. 
 
 The platooning reward of a vehicle that traverses an edge $e$ in a platoon of $n$ vehicles is denoted by $R(n,e)$. It only depends on the edge and the number of vehicles in the platoon, so the platoon members get equal reward in our model for simplicity. This assumption is valid if the platooning benefit is shared equally, via transactions,  within each platoon. The platooning benefit may be, for instance, the reduced fuel consumption and decreased workload of the drivers. The platooning reward on each edge may depend on its length, road gradient information, and other factors that have an impact on the platooning benefit. The reward of each vehicle in the platoon, which is formed at time instance $t$ on edge $e$ is denoted by $R(|\C{e}{t}{\boldsymbol w}{\tau} |,e)$. 
 
  The utility function of each vehicle includes the platooning reward at each edge in its path and the cost of waiting at nodes. The utility function of vehicle $i$ is 
 \begin{align}\label{utility}
 &U^i\left( \actioni{i}, \actioni{-i},\tau \right)= \\ \nonumber & \hspace{1cm}
 \sum \limits_{k=1}^{|\mathcal P^i|}R\left (|\C{\edgeki{i}{k}}{\DT{k}{i}}{\boldsymbol w}{\tau} |, \edgeki{i}{k} \right )   -\Lambda_i(\actioni{i}),
 \end{align}
where $\Lambda_i(\actioni{i})$ denotes vehicle $i$'s cost of waiting at nodes. 

\subsection{Deterministic platoon coordination game}\label{DPMG}

We model the interaction among the vehicles by a non-cooperative game, where each vehicle aims to maximize its own profit. The players of the game are the vehicles in the set $\mathcal N$. The action of each vehicle $i\in \mathcal N$ is its waiting times $\actioni{i}\in \actionspace{i}$. The action space of the game is $\mathcal W=\actionspace{1} \times ...\times \actionspace{N} $, where $\times$ denotes the Cartesian product. The platoon coordination game with deterministic travel times is defined by the triplet $G^d~=~\left( \mathcal N, \mathcal W, \mathcal U^d\right)$, where $\mathcal U^d~=~\{U^i(\actioni{i},\actioni{-i},\tau)\}_i$. We consider a pure NE as the solution concept of the game. 

The theory of potential games is utilized and explored later for showing existence and computation of a pure NE. Potential games were defined and their most important properties shown in \cite{Monder1996}. Congestion games, which were introduced in \cite{Rosenthal1973}, falls under the class of potential games. In \cite{Marden2009} and \cite{Etesami2015}, multi-agent systems were analyzed within the framework of potential games.

\subsection{Existence of NE}
Before we state our result on existence of pure NE, we start by defining the notions of pure NE and exact potential games. A pure NE of the game $G^d$ is an action profile $ \NE{} \in \mathcal  W$ such that, for all $i\in \mathcal N$ and all  $\actioni{i}\in \actionspace{i}$, we have
\begin{equation}
U^i(\NE{i},\NE{-i},\tau)\geq U^i(\actioni{i},\NE{-i},\tau).
\end{equation} 
The game $G^d$ is an exact potential game if there exists a function $\Phi: \mathcal W \rightarrow \mathbb{R}$ such that, for all $i\in \mathcal{N}$ and all $  \actionipotp{i},\actionipotpp{i}\in \actionspace{i} $ and all $ \actioni{-i} \in   \actionspace{-i}$, we have
\begin{align}\label{PotentialProperty} \nonumber
&\Phi(\actionipotp{i},\actioni{-i},\tau)-\Phi( \actionipotpp{i},\actioni{-i},\tau)= \\ & \hspace{3cm} U_i(\actionipotp{i},\actioni{-i},\tau)-U_i( \actionipotpp{i},\actioni{-i},\tau).
\end{align}
 Every finite potential game admits at least one pure NE \cite{Monder1996}.

\begin{theorem}\label{lemma1}
	The deterministic platoon coordination game $G^d$ is an exact potential game with potential function 
	\begin{align}\label{pot}
	\Phi(\boldsymbol{w},\tau) =  \sum \limits_{t\in \mathbb{Z}_+ }  \sum   \limits_{e\in \mathcal E}    r(| \C{e}{t}{\boldsymbol{w}}{\tau} |,e)      -   \sum \limits_{i \in \mathcal N } \Lambda_i( \actioni{i}),
	\end{align}
	where 
	\begin{equation}\label{Eq: Integ}
	r(n,e)=\sum \limits_{j=1}^{n}R(j,e),
	\end{equation}  
	and it thus admits at least one pure NE.
\end{theorem}

\begin{proof}
	See Appendix.
\end{proof}

\subsection{NE seeking algorithm}

In finite potential games, if each player updates its action one at a time according to its best response function, the action profile converges to a pure NE \cite{Monder1996}, \cite{Nisan2007}. The best response function of vehicle $i$, given $\actioni{-i}$, is defined as 
\begin{equation}\label{brf}
B^i (\actioni{-i})=  \underset{\actioni{i}\in  \actionspace{i} }{\arg\max}   \ U^i(\actioni{i}, \actioni{-i},\tau).
\end{equation}
 The best response $B^i (\actioni{-i})$ can simply be computed by looping over all actions in $\actionspace{i}$.  Algorithm \ref{alg} is an NE seeking algorithm based on the best response dynamics. The algorithm converges since the game $G^d$ is a finite potential game. Later, in the simulation study in Section~\ref{sim}, Algorithm \ref{alg} is used to find a pure NE for the vehicles. In practice, The NE of the game can be computed by a trusted computing entity, e.g., a cloud. In this approach, each vehicle transmits the necessary information, e.g., its location, origin, destination, to the computing entity via the Internet and mobile network. The computing entity then computes the NE of the game using the best response algorithm and informs each truck about its NE strategy.

\begin{algorithm} 
	\SetKwInOut{Input}{input}\SetKwInOut{Output}{output}
	\Input{Initial strategy profile, $\boldsymbol w=(\actioni{1},...,\actioni{N})$}
	\Output{NE, $\NE{}  $}
	\BlankLine
$\boldsymbol w^{old}=\boldsymbol w+\Delta, $ where $  \Delta\neq 0$ \\	
	\While{$\boldsymbol w^{old}   \neq  \boldsymbol w^{}  $ }{ $\boldsymbol w^{old}  =\boldsymbol w^{} $ \\ 
		\For{$i \in \mathcal{N}$}{
		
			$\actioni{}   = (\actioni{1},...,\actioni{i-1},B^i (\actioni{-i} ),\actioni{i+1},...,\actioni{N})$
		}
	}
	$\NE{}=\boldsymbol { w}$
	\caption{NE seeking algorithm}
	\label{alg}
\end{algorithm}

\section{Platoon coordination game with stochastic travel times: an open-loop solution } \label{PMGSt1}

In the previous section, we assumed the travel times to be known \textit{a priori}. In this section, we consider the platoon coordination problem when travel times are stochastic. The vehicles that form platoons and therefore vehicles' platooning rewards are not known initially, even if the waiting times of vehicles are given. In this section, we first extend notations to cover stochastic travel times and we define utility functions that include the expected platooning reward and waiting cost. Then, the game that models the platoon coordination scenario with stochastic travel times is presented and we show that it admits a pure NE.

\subsection{Stochastic travel times and expected utility}\label{RSM}

 The stochastic travel time over edge $e$ of vehicles that enter it at time instance $t$ is denoted by the stochastic variable $\RTT{e}{t}$. The start time of vehicle $i$ is denoted by the stochastic variable $\boldsymbol\tau_{0}^i$. We define $\boldsymbol \tau$ to be the set that contains the stochastic variables $\RTT{e}{t}$ and $\boldsymbol\tau_{0}^i$, for all $e\in \mathcal E$, $t\in \mathbb{Z}_+$ and all $ i\in \mathcal N$. The realization of $\RTT{e}{t}$, $\boldsymbol\tau_{0}^i$ and $\boldsymbol \tau$ are denoted by $\TT{e}{t}$, $\tau_{0}^i$ and $\tau$, respectively. The probability of the event $\boldsymbol \tau=\tau$ is denoted $\Pr(\boldsymbol\tau=\tau)$. Given the event $\boldsymbol\tau=\tau$ and the waiting times $\actioni{i}$, the departure times at the nodes in the path of each vehicle~$i$ are calculated by the recursions in \eqref{firstdeparturetime} and \eqref{arrivaltimes}. Therefore, given the event $\boldsymbol\tau=\tau$ and waiting times $\actioni{}$, the set of vehicles that form platoons are indicated by \eqref{PF} and the utility of vehicle $i$ is given by \eqref{utility}.  When the travel times are stochastic, the utility function of vehicle $i$ is defined as $U^i(\actioni{i},\actioni{-i})=~\ES{U^i(\actioni{i},\actioni{-i},\boldsymbol \tau)}$, which can be written 
\begin{equation*}
U^i(\actioni{i},\actioni{-i})=\sum \limits_{\tau} \Pr(\boldsymbol\tau=\tau) U^i(\actioni{i},\actioni{-i}, \tau),
\end{equation*}
which includes the expected platooning reward and cost of waiting.

\subsection{Stochastic platoon coordination game}\label{SPMG}

The non-cooperative game that models the interaction among vehicles when the travel times are stochastic is defined by the triplet $G^s~=~\left( \mathcal N, \mathcal W, \mathcal U^s\right)$, where $  \mathcal U^s~=~\{U^i(\actioni{i},\actioni{-i})\}_i$. Note that the stochastic platoon coordination game $G^s$ and the deterministic platoon coordination game $G^d$ differ in their utility functions, but their sets of players and action spaces are the same.

\subsection{Existence of NE}

The notions of pure NE and exact potential games in case of the stochastic platoon coordination game is similar to the deterministic case. The pure NE of $G^s$ is defined as in Section \ref{DPMG}, but with $U^i(\actioni{i},\actioni{-i})$ instead of $U^i(\actioni{i},\actioni{-i}, \tau)$. The potential game in case of stochastic travel times is defined as in Section~\ref{DPMG}, but with 	$\Phi(\actioni{i},\actioni{-i})$ instead of $\Phi(\actioni{i},\actioni{-i},\tau)$, with obvious modifications.

	\begin{theorem}\label{lemma2}

	The stochastic platoon coordination game $G^s$ is an exact potential game with potential function

	\begin{align}\label{pot3}
	\Phi(\boldsymbol{w}) =  \sum \limits_{\tau} \Pr(\boldsymbol{\tau}=\tau) \Phi(\boldsymbol{w},\tau),
	\end{align}
	where $\Phi(\boldsymbol{w},\tau)$ is given in \eqref{pot}, with $\tau$ being the deterministic travel times and start times of the vehicles. Hence, $G^s$ admits at least one pure NE.
	\end{theorem}


\begin{proof}
	
	See Appendix.
	
\end{proof}

The best response function of the game $G^s$ can be defined as in \eqref{brf}, but with $U^i(\actioni{i},\actioni{-i})$ instead of $U^i(\actioni{i},\actioni{-i}, \tau)$. Then, Algorithm \ref{alg} can be used to find an NE. Algorithm \ref{alg} is guaranteed to converge since $G^s$ is a finite potential game.

\section{Platoon coordination game with stochastic travel times: a feedback solution}\label{PMPS}

In this section, we propose a competitive stochastic decision-making process model for the platoon coordination problem under stochastic travel times. As illustrated in Figure \ref{DMC}, the state of each vehicle $i$ in this model is its location (a node or an edge) and its control input is its waiting times. In this model, the disturbance is the uncertainty of the stochastic travel times on the edges. Two feedback solutions are developed for this problem, where vehicles can update their waiting times at decision-making instances. First, we explain the model and extend notations to cover multiple decision-making instances. Then, the feedback solutions are given.

\subsection{Decision-making instances and experienced travel time}

In the feedback solution, we allow vehicles to update their decisions, at time instances, when at least one vehicle is located at a node. Hence, we allow vehicles to improve their decisions based on the observed information such as the history of travel times and current locations of the vehicles. The $n$th decision-making instance is denoted $t_n^\star\in \mathbb{Z}_+$. Note that decision-making instances are stochastic. The history of travel times up to the $n$th decision-making instance is denoted by the stochastic variable $\boldsymbol{\tau}^{\rm h}_n$ and its realization is denoted by ${\tau}^{\rm h}_n$. The conditional probability of $\boldsymbol{\tau}~=~{\tau}$,  given the history $\boldsymbol{\tau}^{\rm h}_n~=~{\tau}^{\rm h}_n$, is denoted  $\Pr(\boldsymbol{\tau}~=~{\tau}|\boldsymbol{\tau}^{\rm h}_n~=~{\tau}^{\rm h}_n)$. Note that the travel times in $\boldsymbol{\tau}$ that has been observed up to the decision-making instance $t_n^\star$ are known to the vehicles, given the history of travel times ${\tau}^{\rm h}_n$. Let the stochastic variable $\boldsymbol{\tau}^i_{n}$ denote the number of time instances left until vehicle~$i$ arrives to a node in its path, at the decision-making instance $t^\star_n$. The realization of $\boldsymbol{\tau}^i_{n}$ is denoted by ${\tau}^i_{n}$. If vehicle $i$ is located at a node at the decision-making instance $t^\star_n$, then $\tau^i_{n}=0$.

\subsection{Decision variables}

 If vehicle $i$ is located on edge $\edgeki{i}{k}$  at a decision-making instance, then it calculates its waiting times at its $H$ next nodes $\nodeki{i}{k+1},...,\nodeki{i}{k+H}$.  If vehicle $i$ is located at node $\nodeki{i}{k}$ at a decision-making instance, then it calculates its waiting times at nodes $\nodeki{i}{k},...,\nodeki{i}{k+H}$.  The waiting time of vehicle $i$ at node $\nodeki{i}{k}$, calculated at the $n$th decision-making instance, is denoted $\WT{k}{i}{n}$. The waiting times calculated by vehicle $i$, at the decision-making instance $t_n^\star$, are denoted $\AWT{i}{n}=(\WT{k+1}{i}{n},...,\WT{k+H}{i}{n}  )$ or $\AWT{i}{n}=(\WT{k}{i}{n},...,\WT{k+H}{i}{n} )$, depending on whether it is located on edge $\edgeki{i}{k}$ or at node $\nodeki{i}{k}$. The waiting times of all vehicles are denoted $\boldsymbol{w}_n$. The action space of vehicle $i$ at the $n$th decision-making instance is denoted $\mathcal W^i_n$. The action space is updated at each decision-making instance according to the history of actions. The action space of all vehicles, at the decision-making instance $t_n^\star$, is denoted by $\mathcal W_n$.

 Vehicles that are located at nodes at a decision-making instance leave their nodes if their calculated waiting times are zero. That is, if vehicle $i$ is located at $\nodeki{i}{k}$ at the $n$th decision-making instance, then it leaves node $\nodeki{i}{k}$ if its calculated waiting time $\WT{k}{i}{n}=0$. Vehicle $i$ stays at node $\nodeki{i}{k}$ at the $n$th decision-making if its calculated waiting time $\WT{k}{i}{n}~>~0$. If vehicle $i$ stays at $\nodeki{i}{k}$, it will have a chance to update its waiting time there since it triggers decision-making in the next time instance. 
 

	\begin{figure}
	\centering
	\includegraphics[scale=0.88]{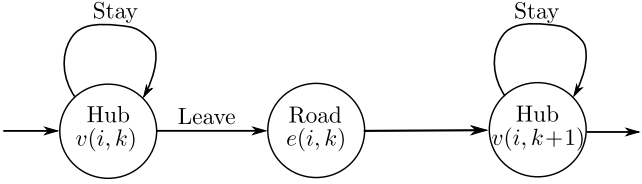}
	\caption{A state machine illustrating the decision-making process. When vehicle $i$ is located at a node it decides whether to stay at the node or leave. If vehicle $i$ decides to stay, it remains at the node until the next time instance when it decides again whether to stay or leave. If vehicle $i$ decides to leave, it departs from the node and enters its next edge. This process is repeated for all nodes and edges in the path of vehicle $i$. }
		
	\label{DMC}
\end{figure}

\subsection{Departure times and utility functions}

Let $\DDT{k}{i}{n}$ denote vehicle $i$'s departure time at node $\nodeki{i}{k}$ calculated at the $n$th decision-making instance. Given $\boldsymbol \tau=~\tau$ and $ \AWT{i}{n}$, if vehicle $i$ is located on edge $\edgeki{i}{k}$,  then its departure time from node $\nodeki{i}{k+1} $ is computed as   
\begin{align}\label{SDDT}
\DDT{k+1}{i}{n}=t_n^\star+ \tau^i_{n}+\WT{k+1}{i}{n},
\end{align}
where $\tau_{n}^i$ is the number of time instances left until it arrives to node $\nodeki{i}{k+1}$. If vehicle $i$ is located at node $\nodeki{i}{k+1}$ at the $n$th decision-making instance, then the departure time from $\nodeki{i}{k+1}$ is computed by equation \eqref{SDDT} with $\tau_{n}^i=0$. The departure times from the remaining nodes are computed as 
\begin{align}\label{SDDTp}
\DDT{l+1}{i}{n}=\DDT{l}{i}{n}
+\TT{\edgeki{i}{l}}{\DDT{l}{i}{n}}+\WT{l+1}{i}{n},
\end{align}
where $l>k$. 

Given $\boldsymbol \tau=\tau$ and $ \AWT{i}{n}$, the set of vehicles that enter edge $e$ at time instance $t$ is denoted by $\MC{e}{t}{\boldsymbol w_n}{ \tau}{n} $, and if vehicle $i$ is located on edge $\edgeki{i}{j}$ at the $n$th decision-making instance, its utility function is 
\begin{align}\label{utilitydrhs}
& U_{n}^i\left( \AWT{i}{n}, \AWT{-i}{n},  \tau \right)= \\ \nonumber  & \hspace{0.4cm} \sum \limits_{k=j+1}^{j+H}R\Big (|\MC{\edgeki{i}{k}}{\DDT{k}{i}{n}}{\boldsymbol w_n} { \tau}{n} |, \edgeki{i}{k} \Big )    -\Lambda_i(\AWT{i}{n}),
\end{align}
where $\Lambda_i(\AWT{i}{n})$ is the cost of waiting for its entire trip. Here, the dependency on the waiting times at nodes outside the horizon is not indicated explicitly.

\subsection{Deterministic receding horizon solution }\label{DRHS}

Under this solution, the vehicles compute their waiting times at their remaining waiting nodes using the conditional mean of the travel times on edges. The mean is rounded to the nearest integer since the travel times are required to be integer-valued. The conditional (rounded) mean of $\boldsymbol{\tau}$ given the experienced travel times  $\boldsymbol{\tau}^{\rm h}_n={\tau}^{\rm h}_n$, is denoted $\bar \tau_n$. The departure times of vehicles are computed by the recursions \eqref{SDDT} and \eqref{SDDTp} with travel times given by $\boldsymbol \tau_n~=~\bar\tau_n$. Each vehicle aims to maximize its own utility function $U_{n}^i\left( \AWT{i}{n}, \AWT{-i}{n},  \bar\tau_n \right)$, which is given in \eqref{utilitydrhs}. The solution is an NE of the non-cooperative game defined by  $G^d_n~=~\left(\mathcal N, \mathcal W_n, \mathcal U_n^d \right)$, where $\mathcal U_n^d ~=~\{	U_{n}^i\left(\AWT{i}{n}, \AWT{-i}{n}, \bar \tau_n \right)\}_i $.

\begin{corollary}
  The game $G^d_n$ is an exact potential game and thus admits at least one NE.
\end{corollary}

 \begin{proof} The game $G_n^d$ is a deterministic platoon coordination game, so the result follows from Theorem \ref{lemma1}. 
\end{proof}

\subsection{Stochastic receding horizon solution }\label{SRHS}

Under this solution, the vehicles compute their waiting times at their remaining waiting nodes using the conditional distribution of the stochastic travel times. Given $\boldsymbol \tau=\tau$ and $\boldsymbol w_n$, the departure times at nodes of vehicles are computed as \eqref{SDDT} and \eqref{SDDTp}, and the individual utility of vehicle $i$ is $U_{n}^i\left(\AWT{i}{n}, \AWT{-i}{n},  \tau \right)$, as in \eqref{utilitydrhs}. When the travel times are stochastic, each vehicle $i$ aims to maximize its expected individual utility, \emph{i.e.,} 
\begin{equation*}
U_{n}^i(\AWT{i}{n},\AWT{-i}{n})=\CES{U_{n}^i(\AWT{i}{n}, \AWT{-i}{n},\boldsymbol \tau)}{\boldsymbol{\tau}^{\rm h}_n={\tau}^{\rm h}_n},
\end{equation*}
where the conditional probabilities of the travel times are used to compute the expected utility.
The solution is an NE of the non-cooperative game defined by $G^s_n~=~\left(\mathcal N, \mathcal W_n, \mathcal U_n^s \right)$, where $\mathcal U_n^s=\{	U_{n}^i\left(\AWT{i}{n},\AWT{-i}{n} \right) \}_i $.

\begin{corollary}
	The game $G^s_n$ is an exact potential game and thus admits at least one NE.
\end{corollary}

\begin{proof} The game $G_n^s$ is a stochastic platoon coordination game, so the result follows from Theorem \ref{lemma2}. 
\end{proof}

\section{Simulation of the Swedish road network}\label{sim}

In this section, we perform a simulation study over the Swedish road network to evaluate the proposed solutions of the stochastic platoon coordination problem. After explaining the setup of the simulation, we first study the platooning rate and average waiting time of vehicles, as a function of the number of vehicles injected into the network. Then, we inject vehicles into the network at two different time periods and study the number of followers as a function of time. Finally, we vary the waiting time budget and the platooning benefit to see how it affects the platooning rate and the average waiting times of vehicles.

\subsection{Simulation setup}\label{setup}

We consider the graph in Figure \ref{MAP} that represents the Swedish road network. The nodes in the graph represent hubs near cities, where vehicles can wait for others in order to platoon. The edges represent roads that connect hubs. The graph in Figure \ref{MAP} has $34$ nodes and $55$ edges. The origin of each vehicle $i$ is drawn randomly from the set of nodes in the road network and the probability that a node be an origin is proportional to the population of the closest city to the node. The area of the nodes in Figure \ref{MAP} are proportional to their city populations. Then, the destination is drawn randomly from the set of nodes that fulfill that the shortest path between the origin and destination is more then $300$ km and less then $800$ km. For the nodes that fulfill the destination criteria, the probability that a node be a destination is also proportional to the size of the closest city to the node. The numerical results are obtained by Monte Carlo simulations with $50$ samples. For each sample, new origins, destinations and start times of vehicles are randomly generated.

We consider stochastic and time-varying travel times on the edges. Let $\underline \tau(e)$ denote the travel time on edge $e$ at off-peak periods, which is assumed to be deterministic. Then, the travel time of vehicles on edge $e$ that enters at time instance $t$ is $\tau(e,t)=\underline \tau(e) +\Delta(e,t)$, where $\Delta(e,t)$ is the stochastic time-delay (in number of time instances) that occurs at peak periods near urban areas. The time step length is $5$~minutes. For each edge $e$, the time-delay $\Delta(e,t)$ is drawn randomly from one of the time-delays at days 1--10 in Figure \ref{TT}, with equal probability. In the feedback solutions, at each decision-making instance, we compute the conditional expected utilities and the conditional mean of the travel times by excluding the time-delays for which trucks on links would already have arrived at their next nodes.

We assume that leaders have zero benefit from platooning, followers have equal benefit and the benefit is shared equally between platoon members. The platooning reward is defined as $R(n,e)~=~c_b l(e) (n-1)/ n$, where $l(e)$ is the length of edge $e$ and $c_b$ is the platooning benefit per kilometer. Furthermore, the waiting cost of vehicle $i$ is $\Lambda_i(\actioni{i})=c_t( \waitingtime{i}{1}+...+\waitingtime{i}{|\mathcal P^i|})$, where $c_t$ is the cost of waiting per time instance.  We assume $c_b=1.7$ SEK (Swedish crowns) and $c_t=22$ SEK. This is reasonable if the followers save $10 \%$ of fuel and their driver cost is reduced by approximately one-third. We will also vary the platooning benefit $c_b$ to study its impact on the overall performance.

 The waiting budget of each vehicle is $20$ minutes ($4$ time instances) along its whole journey. That is, the initial action space of vehicle $i$ is $\mathcal{W}^i=\{\actioni{i}|\waitingtime{i}{1}+...+\waitingtime{i}{|\mathcal P^i|}\leq 4 \}$. The action space of the vehicles in the feedback solutions are updated so that vehicles do not violate their initial waiting budget.

When the vehicles are many, the feedback solutions are slow to compute. One reason is that many computations are required when calculating a pure NE of many vehicles. Another reason is that decision-making is triggered at almost every time instance when vehicles are many. To speed up the computation of the feedback solutions, only vehicles that would arrive to a node within $20$ minutes, if they were travel with free-flow speed, update their waiting times. The horizon length of the feedback solutions is $H=2$ (the two next nodes).

\begin{figure}[t!]
	\centering
	\includegraphics[scale=0.4]{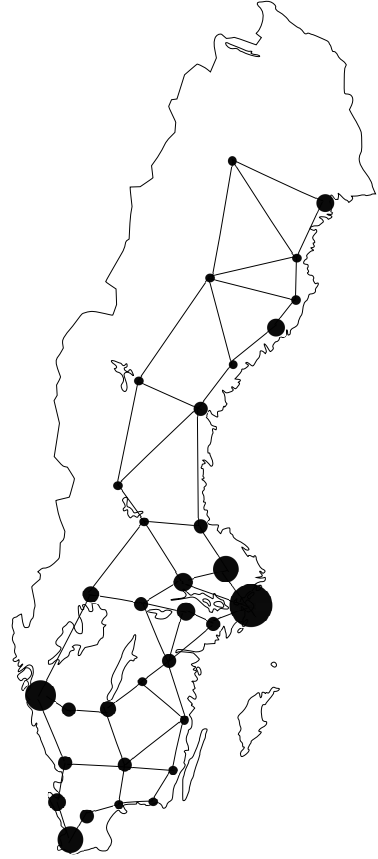}
	\caption{Map of Sweden and the considered road network.}
		\label{MAP}	\end{figure}

\subsection{Evaluation}\label{sim1}

We compare the proposed solutions of the stochastic platoon coordination problem with two additional cases: (1) Vehicles do not wait at nodes, but platoon spontaneously if they enter an edge at the same time instance, (2) vehicles have non-causal information about their future travel times (the solution is an NE of the deterministic platoon coordination game in Section \ref{DPMGS}). The results presented next are obtained by vehicles randomly entering the network between 6:30 and 8:30 a.m. 

The platooning rate and the average used waiting time are shown in Figure \ref{COMPMETHprop}. Here SRHS and DRHS denote the stochastic and deterministic receding horizon solution, respectively. The platooning rate is the ratio between the total followed distance and the total traveled distance. 

Figure \ref{COMPMETHprop1} shows that the platooning rate increases with the number of vehicles in the network. Moreover, the platooning rate is around $5\%$ higher for the stochastic receding horizon solution than for the deterministic receding horizon solution and the difference is smaller when the vehicles in the network are many. The difference in platooning rate of the receding horizon solutions and when the travel times on edges are known is less than $10\%$. The receding horizon solutions had a significantly higher platooning rate than when the vehicles only planned initially and when the vehicles only platoon spontaneously. Furthermore, Figure \ref{COMPMETHprop1} shows that non-cooperative platooning can have significant benefits on a societal scale. For example, the feedback solutions have a platooning rate of approximately $40 \% $ when $1000$  vehicles are considered in the Swedish transportation network. This corresponds to a reduction of $4\%$ of the overall fuel consumption if each follower vehicle save  $10 \% $ of fuel.

Figure \ref{COMPMETHprop2t} shows that for all proposed solutions, the waiting times of the vehicles increases up to a point where it then decreases. This is because, when vehicles are few, the vehicles have few platooning opportunities within their time windows and it might therefore be more beneficial to leave immediately without waiting for others, and when vehicles are many, vehicles does not need to wait for long in order to form platoons with others. 

In Table \ref{tableutility}, we show the total utility of the vehicles when the number of vehicles is varied. Here, KTT, IP and SP denote known travel times, initial planning and spontaneous platooning, respectively. The table shows that the feedback solutions obtain a higher total utility than the initial planning and the spontaneous solution. The table also shows that the highest total utility is achieved when the travel times are known \textit{a priori}.

\begin{figure}
	\centering

	\subfigure[Platooning rate]{\includegraphics[scale=0.62]{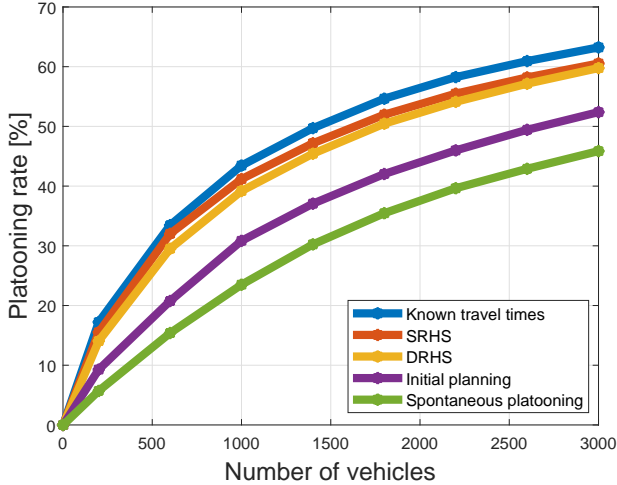}\label{COMPMETHprop1}}
	\vfill
	\subfigure[Average waiting time per vehicle.]{\includegraphics[scale=0.62]{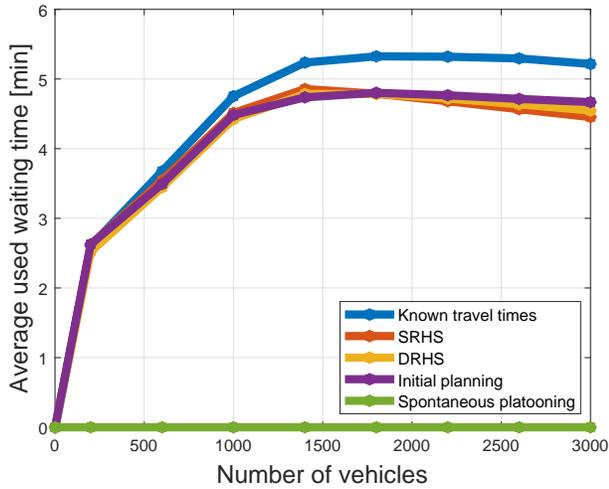}\label{COMPMETHprop2t}}
	\caption{Platooning rate and average used waiting time as a function of number of vehicles injected into the network. SRHS and DRHS stand for stochastic and deterministic receding horizon solution, respectively. }
	\label{COMPMETHprop}
	
\end{figure}

\setlength{\tabcolsep}{10pt}
\renewcommand{\arraystretch}{1.5}

\begin{table}[]
	\centering
	\caption{Total utility [Million SEK] }
	\label{tableutility}
	\begin{tabu}{r|[1.5pt] c|c|c|c|l|[1.5pt]}
		\cline{2-6}
		\multicolumn{1}{l|[1.5pt]}{} & \multicolumn{5}{c|[1.5pt]}{Number of vehicles} \\ \cline{2-6} 
		& \textbf{600} & \textbf{1000} & \textbf{1800} & \textbf{2200} & \multicolumn{1}{c|[1.5pt]}{\textbf{3000}} \\ \specialrule{.2em}{0em}{0em} 
		\multicolumn{1}{|r|[1.5pt]}{KTT} & 0.159 & 0.344 & 0.778 & 1.015 & 1.502 \\ \hline
		\multicolumn{1}{|r|[1.5pt]}{SRHS} & 0.152 & 0.326 & 0.741 & 0.966 & 1.438 \\ \hline
		\multicolumn{1}{|r|[1.5pt]}{DRHS} & 0.140 & 0.310 & 0.719 & 0.94 & 1.420 \\ \hline
		\multicolumn{1}{|r|[1.5pt]}{IP} & 0.010 & 0.244 & 0.599 & 0.801 & 1.245 \\ \hline
		\multicolumn{1}{|r|[1.5pt]}{SP} & 0.073 & 0.186 & 0.506 & 0.691 & 1.089 \\ \specialrule{.2em}{0em}{0em} 
	\end{tabu}
\end{table}

\subsection{Impact of starting times}\label{sim2}

In Figure \ref{COMPMETHtime}, the number of followers during the day is shown when the start times of vehicles are randomly distributed in the intervals 6:30--8:30 a.m. and 4:30--6:30 a.m., respectively. In both cases, the number of vehicles is fixed to $N=1000$. The number of followers increases during the periods when vehicles start their journeys. Figure  \ref{COMPMETHtime1} shows that the number of followers is significantly higher for the feedback solutions than for the solution where the vehicles only planned initially. This is because in the feedback solutions, vehicles update their waiting times and are therefore more likely to form platoons along their journeys. It is observed that the differences between the solutions are much smaller when the vehicles' start times are before the peak period than when the vehicles' start times are during the peak period. This is because when vehicles start their journeys before the peak period, platoons are formed without being exposed to uncertainty in travel times, and platoons remain intact during the peak period.

\begin{figure}[t!]
	\centering

	\subfigure[Vehicles are injected into the network 6:30--8:30 a.m. ]{\includegraphics[scale=0.63]{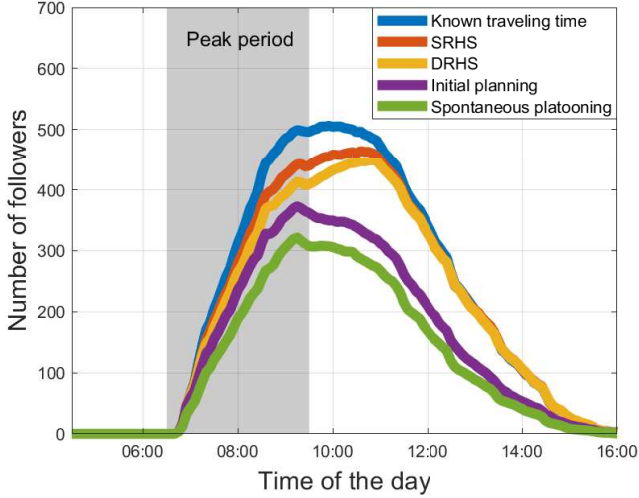}\label{COMPMETHtime1}}
	\vfill
	\subfigure[Vehicles are injected into the network 4:30--6:30 a.m.]{\includegraphics[scale=0.63]{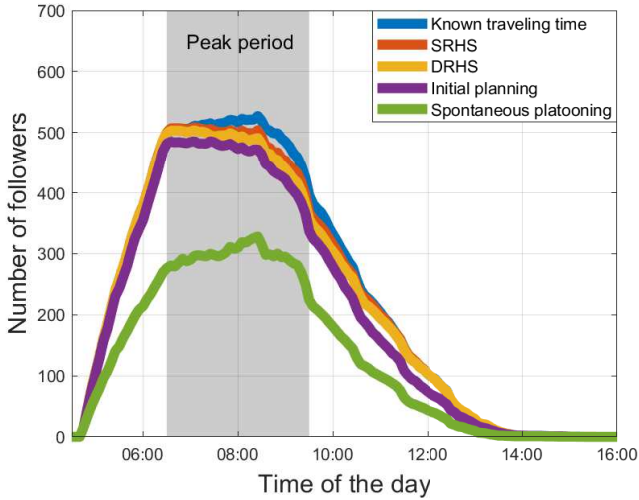}\label{COMPMETHtime2}}
	\caption{Number of followers during two days with different injection periods.}
	\label{COMPMETHtime}
	
\end{figure}

Figure \ref{sizedist} shows the platoon length distribution as a function of time when vehicles are injected into the network in the interval 6:30--8:30 a.m. The number of vehicles is fixed to $N=1000$ and the deterministic receding horizon is used. The figure shows that the most common platoon lengths are $2$ and $3$ vehicles, and that less than $5\%$ of the vehicles drive in platoons of  $7$ vehicles or more.

\begin{figure}[t!]
	\centering
	\includegraphics[scale=0.65]{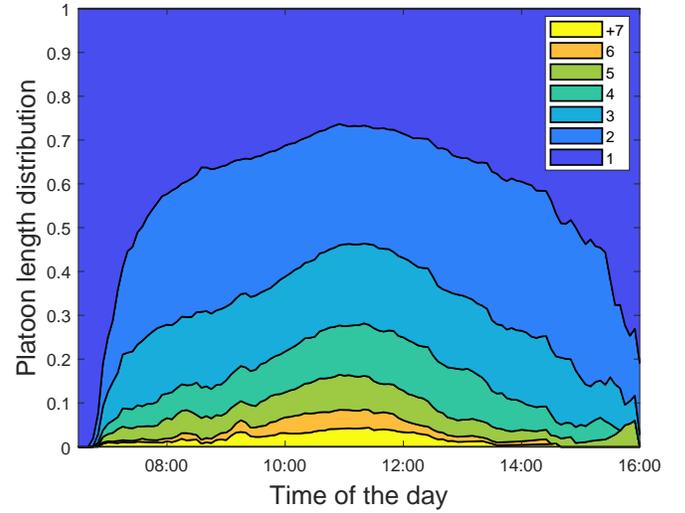}
	\caption{Platoon length distribution of the deterministic receding horizon.}
	\label{sizedist}
\end{figure}


\subsection{Impact of waiting time budget and benefit from platooning}\label{sim3}

In Table \ref{table1} and \ref{table2}, the platooning rate and average waiting time per vehicle are shown, respectively, when the benefit from platooning and the waiting budget are varied. We use three values of the platooning benefit $c_b$. The platooning benefit $c_b=0.5$ SEK represents the case when the platooning benefit is only due to the reduced fuel consumption. The platooning benefits $c_b=1.7$ SEK and $c_b=4$ SEK, represent the cases when, beyond the fuel saving, the driver cost is reduced with approximately one-third and the driver cost is eliminated, respectively. Moreover, the cost of waiting is kept fixed to $c_t=22$ SEK, $1000$ vehicles was injected into the network and the deterministic receding horizon solution was used. Table \ref{table1} shows that the platooning rate increases with the platooning benefit as well as the vehicles' waiting budget. Table~\ref{table2} shows that the average waiting time per vehicle increases when the waiting budget increases and when the platooning benefit increases.

\setlength{\tabcolsep}{10pt}
\renewcommand{\arraystretch}{1.5}


\begin{table}[]
	\centering
	\caption{Platooning rate [$\%$] }
	\label{table1}
	\begin{tabu}{ cr |[1.5pt] c  | c|c|c|l|[1.5pt]}
		\cline{3-7}
		\multicolumn{1}{l}{} & \multicolumn{1}{l|[1.5pt]}{} & \multicolumn{5}{c|[1.5pt]}{Waiting time budget [min]} \\ \cline{3-7} 
		\multicolumn{1}{l}{} &  & \textbf{5} & \textbf{10} & \textbf{15} & \textbf{20} & \textbf{25} \\ \specialrule{.2em}{0em}{0em}
		\multicolumn{1}{|c|}{} & \textbf{4} & 31.0 & 35.5 & 37.4 & 38.8 & 38.9 \\ \cline{2-7} 
		\multicolumn{1}{|c|}{$c_b$} & \textbf{1.7} & 30.9 & 35.3 & 37.2 & 38.2 & 38.3 \\ \cline{2-7} 
		\multicolumn{1}{|c|}{} & \textbf{0.5} & 29.8 & 31.6 & 32.1 & 32.1 & 32.1 \\
	 \specialrule{.2em}{0em}{0em} 
	\end{tabu}
\vspace{0.7cm}
		\centering
	\caption{Waiting time per vehicle [minutes]}
	\label{table2}
	\begin{tabu}{ cr |[1.5pt] c  | c|c|c|l|[1.5pt]}
		\cline{3-7}
		\multicolumn{1}{l}{} & \multicolumn{1}{l|[1.5pt]}{} & \multicolumn{5}{c|[1.5pt]}{Waiting time budget [min]} \\ \cline{3-7} 
		\multicolumn{1}{l}{} &  & \textbf{5} & \textbf{10} & \textbf{15} & \textbf{20} & \textbf{25} \\ \specialrule{.25em}{0em}{0em}
		\multicolumn{1}{|c|}{} & \textbf{4} & 1.97 & 3.11 & 3.64 & 4.01 & 4.01 \\ \cline{2-7} 
		\multicolumn{1}{|c|}{$c_b$} & \textbf{1.7} & 1.92 & 3.02 & 3.48 & 3.70 & 3.71 \\ \cline{2-7} 
		\multicolumn{1}{|c|}{} & \textbf{0.5} & 1.18 & 1.55 & 1.60 & 1.60 & 1.60 \\ 
\specialrule{.2em}{0em}{0em}
	\end{tabu}
\end{table}

\section{Conclusions}\label{con}

The platoon coordination problem, where vehicles can wait at hubs in a transportation network in order to platoon with others was considered in this paper. The strategic interaction among the vehicles when deciding on their waiting times at hubs was formulated as a game. Two models were developed: one with deterministic travel times and one with stochastic travel times. We have shown that both games are potential games and therefore admit at least one pure NE. The pure NEs of the games form open-loop solutions, where vehicles calculate their waiting times at the initial time instance. In the case of stochastic travel times, we proposed two feedback solutions, where vehicles update their waiting times along their journeys.

The proposed solutions have been evaluated in a simulation study over the Swedish road network, where the travel times are stochastic, time-varying, and generated from real data. It was shown in the simulation study that uncertainty in the travel times had a large impact on the platooning rate and  the feedback solutions had a significantly higher platooning rate than the open-loop solution. The platooning rate of the feedback solutions was almost as high as in the case where vehicles had ideal non-causal information of their future travel times. Furthermore, the simulation study showed that significant benefits from platooning can be obtained at a societal scale, even when vehicles aim to optimize their individual profits.

We also studied the total number of platoon followers, as a function of the time of the day, when vehicles were injected into the network during and before the peak-period. The number of platoon followers was much higher when the vehicles started their journeys before the peak period, then many platoons were formed before the peak period and the platoon formations remained intact during the peak period.

As an avenue of our future research, we will consider the equilibrium analysis under incomplete information and under more realistic cost, benefit and travel models. We will also investigate the subgame perfect equilibrium as a solution concept of a non-cooperative platoon coordination problem with stochastic travel times.

\section*{Acknowledgments}
 We thank Erik Jenelius for providing travel time data.

\appendix

\begin{proof}[Proof of Theorem \ref{lemma1}]
We show that property \eqref{PotentialProperty} holds for the candidate potential function in \eqref{pot}. Consider two feasible actions of vehicle $i$ denoted by
 $\actionipotp{i}=(\waitingtimep{i}{1},...,\waitingtimep{i}{|\mathcal P^i|})$ and $\actionipotpp{i}=(\waitingtimepp{i}{1},...,\waitingtimepp{i}{|\mathcal P^i|})$. The action profiles corresponding to $\actionipotp{i}$ and $\actionipotpp{i}$ are denoted by $\actionipotp{}=(\actionipotp{i},\actioni{-i})$ and $\actionipotpp{}=(\actionipotpp{i},\actioni{-i})$, respectively, for an arbitrary $\actioni{-i}$. Given $\tau$, the departure time of vehicle $i$ from node $\nodeki{i}{k}$ is $\DTp{i}{k}$ under action $\actionipotp{i}$ and $ \DTpp{i}{k}$ under action $\actionipotpp{i}$. Let $\mathcal P^{i}_d$ be the collection of edges with uncommon entering times under the actions $\actionipotp{i}$ and $\actionipotpp{i}$,  \emph i.e., $\mathcal P^{i}_d=\{\edgeki{i}{k}\in \mathcal P^{i}| \DTp{i}{k}\neq \DTpp{i}{k} \}$. Then, $\Phi (\actionipotp{i},\actioni{-i},\tau)$ and $\Phi (\actionipotpp{i},\actioni{-i},\tau)$ can be written as \eqref{Proof1} and \eqref{Proof2}, respectively, where $\widehat{O}_i$ and $\widecheck{O}_i$ denote the remaining terms in the summations. Note that, the terms in $\widehat{O}_i$ ($\widecheck{O}_i$) correspond to the travel times and edges which are not affected by changing the decision of vehicle $i$ from $\actionipotp{i}$ to $\actionipotpp{i}$. Thus, we have $\widehat{O}_i=\widecheck{O}_i$. Using this fact, the difference between $\Phi (\actionipotp{i},\actioni{-i},\tau)$ and $\Phi(\actionipotpp{i},\actioni{-i},\tau)$  can be expressed as  \eqref{Proof4}. Similarly, the difference between the utility of vehicle $i$ under  $\actionipotp{i}$ and $\actionipotpp{i}$ can be  written as \eqref{Proof5}. Under $\actionipotp{i}$, vehicle $i$ leaves node $\nodeki{i}{k}$ and enters edge  $\edgeki{i}{k} \in\mathcal{P}_d^{i} $ at time $\DTp{i}{k}$. Since $\DTp{i}{k}$ is different from $\DTpp{i}{k}$ when $\edgeki{i}{k}\in\mathcal{P}_d^{i}$, vehicle $i$ will not be part of the platoon which might form at node $\nodeki{i}{k}$ at time $\DTpp{i}{k}$. Thus, we have 
\begin{align}\label{Eq: Card1}
|\C{\edgeki{i}{k}}{\DTp{i}{k}}{\actionipotp{}}{\tau} |=|\C{\edgeki{i}{k}}{\DTp{i}{k}}{\actionipotpp{}}{\tau} |+1
\end{align}
for all $k$ such that $\edgeki{i}{k}\in\mathcal{P}_d^{i}$ since only vehicle $i$ changes its strategy. Similarly, we have 
\begin{align}\label{Eq: Card2}
|\C{\edgeki{i}{k}}{\DTpp{i}{k}}{\actionipotpp{}}{\tau} |=|\C{\edgeki{i}{k}}{\DTpp{i}{k}}{\actionipotp{}}{\tau} |+1
\end{align}
for all $k$ such that $\edgeki{i}{k}\in\mathcal{P}_d^{i}$. From \eqref{Eq: Integ}, we have 
\begin{align}\label{Eq: Recu}
r\left(n+1,e\right)-r\left(n,e\right)=R\left(n+1,e\right)
\end{align}
for all $n\in \mathbb Z_+$. Using \eqref{Eq: Card1}-\eqref{Eq: Recu}, we have \footnotesize
\begin{align*}
& R \big(|\C{\edgeki{i}{k}}{\DTp{i}{k}}{\actionipotp{}}{\tau} |,\edgeki{i}{k}\big)= \\ & \hspace{0.7cm} r(|\C{\edgeki{i}{k}}{\DTp{i}{k}}{\actionipotp{}}{\tau} |,\edgeki{i}{k}\big)- r(|\C{\edgeki{i}{k}}{\DTp{i}{k}}{\actionipotpp{}}{\tau} |,\edgeki{i}{k}\big)\nonumber
\end{align*}\normalsize
and \footnotesize
\begin{align*}
& R \big(|\C{\edgeki{i}{k}}{\DTpp{i}{k}}{\actionipotpp{}}{\tau} |,\edgeki{i}{k}\big)= \\ & \hspace{0.7cm} r(|\C{\edgeki{i}{k}}{\DTpp{i}{k}}{\actionipotpp{}}{\tau} |,\edgeki{i}{k}\big)- r(|\C{\edgeki{i}{k}}{\DTpp{i}{k}}{\actionipotp{}}{\tau} |,\edgeki{i}{k}\big)\nonumber.
\end{align*} \normalsize
It follows that the game is an exact potential game by the above equations and equation \eqref{Proof6}. The game thus admits at least one NE \cite{Monder1996}.

\begin{figure*}
	\small
\begin{align} 
	\Phi(\actionipotp{i},\actioni{-i},\tau)&=  \sum\limits_{\edgeki{i}{k} \in \mathcal P_d^{i}} \left[ 
r(|\C{\edgeki{i}{k}}{\DTp{i}{k}}{\actionipotp{}}{\tau} |,\edgeki{i}{k}\big)+r(|\C{\edgeki{i}{k}}{\DTpp{i}{k}}{\actionipotp{}}{\tau} |,\edgeki{i}{k}\big)\right]  -\Lambda_i(\actionipotp{i}) + \widehat{O}_i \label{Proof1} \\
	\Phi(\actionipotpp{i},\actioni{-i},\tau)&= \sum\limits_{\edgeki{i}{k} \in \mathcal P_d^{i}}  \left[ r(|\C{\edgeki{i}{k}}{\DTp{i}{k}}{\actionipotpp{}}{\tau} |,\edgeki{i}{k}\big)  +  r(|\C{\edgeki{i}{k}}{\DTpp{i}{k}}{\actionipotpp{}}{\tau} |,\edgeki{i}{k}\big) \right]  -\Lambda_i( \actionipotpp{i})+\widecheck{O}_i  \label{Proof2}
	\end{align}
	\vspace{-0.2cm}
	\hrule
	\begin{align}
	\Phi(\actionipotp{i},\actioni{-i},\tau)-	\Phi(\actionipotpp{i},\actioni{-i},\tau) &=-\Lambda_i(\actionipotp{i})    +\Lambda_i(\actionipotpp{i}) \nonumber\\
	&+\sum\limits_{\edgeki{i}{k} \in \mathcal P_d^{i}}  \bigg(  \left[r(|\C{\edgeki{i}{k}}{\DTp{i}{k}}{\actionipotp{}}{\tau}|,\edgeki{i}{k}\big)-r\big(|\C{\edgeki{i}{k}}{\DTp{i}{k}}{\actionipotpp{}}{\tau}|,\edgeki{i}{k}\big) \right] \nonumber\\
	&\hspace{3cm}-\left[ r(|\C{\edgeki{i}{k}}{\DTpp{i}{k}}{\actionipotpp{}}{\tau}|,\edgeki{i}{k}\big) -  r(|\C{\edgeki{i}{k}}{\DTpp{i}{k}}{\actionipotp{}}{\tau}|,\edgeki{i}{k}\big)\right] \bigg )\label{Proof4}        
	\end{align} \vspace{-0.2cm}
	\hrule
	\begin{align}\label{Proof5}
	& U^i(\actionipotp{i},\actioni{-i},\tau)-U^i(\actionipotpp{i},\actioni{-i},\tau)=  -\Lambda_i(\actionipotp{i})   +\Lambda_i(\actionipotpp{i}) + \!\!\!\!\!\!\sum\limits_{\edgeki{i}{k} \in \mathcal P_d^{i}} \Big (  R \big(|\C{\edgeki{i}{k}}{\DTp{i}{k}}{\actionipotp{}}{\tau}|,\edgeki{i}{k}\big) -  R \big(|\C{\edgeki{i}{k}}{\DTpp{i}{k}}{\actionipotpp{}}{\tau}|,\edgeki{i}{k}\big) \Big )
	\end{align} \vspace{-0.2cm}
	\hrule
	\begin{align}\nonumber\label{Proof6}
	& \Phi(\actionipotp{i},\actioni{-i},\tau)-\Phi(\actionipotpp{i},\actioni{-i},\tau)-U^i(\actionipotp{i},\actioni{-i},\tau)+U^i(\actionipotpp{i},\actioni{-i},\tau)= \\
	&\hspace{1cm} \sum\limits_{\edgeki{i}{k} \in \mathcal P_d^{i}}   \bigg(\left[ r(|\C{\edgeki{i}{k}}{\DTp{i}{k}}{\actionipotp{}}{\tau}|,\edgeki{i}{k}\big)-r\big(|\C{\edgeki{i}{k}}{\DTp{i}{k}}{\actionipotpp{}}{\tau}|,\edgeki{i}{k}\big)-R \big(|\C{\edgeki{i}{k}}{\DTp{i}{k}}{\actionipotp{}}{\tau}|,\edgeki{i}{k}\big)\right] \\ \nonumber
	&\hspace{3cm} + \left[r(|\C{\edgeki{i}{k}}{\DTpp{i}{k}}{\actionipotp{}}{\tau}|,\edgeki{i}{k}\big) -  r(|\C{\edgeki{i}{k}}{\DTpp{i}{k}}{\actionipotpp{}}{\tau}|,\edgeki{i}{k}\big) +  
	 R \big(|\C{\edgeki{i}{k}}{\DTpp{i}{k}}{\actionipotpp{}}{\tau}|,\edgeki{i}{k}\big)\right] \bigg)  
	\end{align}
	\vspace{-0.2cm}
	\hrule
	\normalsize
\end{figure*}
\end{proof}

\begin{proof}[Proof of Theorem \ref{lemma2}]
We show that property \eqref{PotentialProperty} holds for the candidate potential function in \eqref{pot3}. Consider two feasible actions of vehicle $i$ denoted by $\actionipotp{}=(\actionipotp{i},\actioni{-i})$ and $\actionipotpp{}=(\actionipotpp{i},\actioni{-i})$, respectively, for an arbitrary $\actioni{-i}$. By the result in Theorem \ref{lemma1}, for a given $\tau$, for all $i\in \mathcal N$, all $\actionipotp{}$,$\actionipotpp{} \in \mathcal W$, we have  
\begin{align*} 
\Phi (\actionipotp{},\tau)-\Phi (\actionipotpp{},\tau)= 
U^{i}\left(\actionipotp{},\tau\right)-U^{i}\left(\actionipotpp{},\tau \right).
\end{align*}
Furthermore, we have \small
\begin{align*}
U^{i}\left(\actionipotp{}\right)-U^{i}\left(\actionipotpp{}\right)= 
\sum \limits_{\tau} \Pr(\boldsymbol{\tau}=\tau)	\left(U^{i}\left(\actionipotp{}, \tau \right)-U^{i}\left(\actionipotpp{}, \tau  \right) \right)
\end{align*} \normalsize
and  \small
\begin{align*}
\Phi\left(\actionipotp{} \right)-\Phi\left(\actionipotpp{}\right)= 
\sum \limits_{\tau} \Pr(\boldsymbol{\tau}=\tau)	\left(\Phi\left(\actionipotp{},\tau \right)-\Phi\left(\actionipotpp{},\tau \right) \right).
\end{align*}		\normalsize	
It follows that the game is an exact potential game by the three equations above. It thus admits at least one NE \cite{Monder1996}. 	
\end{proof}

\bibliographystyle{ieeetr} %
\bibliography{RefDatabase}

\end{document}